\documentclass[12pt]{amsart}

\usepackage{amsmath,amssymb,lineno,amsthm,fullpage,parskip,graphicx}

\newtheorem{theorem}{Theorem}
\newtheorem{lemma}[theorem]{Lemma}

\title{Colourings of Oriented Connected Cubic Graphs}

\begin{document}

\maketitle
\author{Christopher Duffy\footnote{Research supported by the Natural Science and Research Council of Canada}}\\
\address{Department of Mathematics and Statistics, University of Saskatchewan, CANADA}
\begin{abstract}
In this  note we show  every orientation of a connected  cubic graph admits an oriented $8$-colouring. This lowers the best-known upper bound for the chromatic number of the family of orientations of connected cubic graphs.
We further show that every such oriented graph admits a $2$-dipath $7$-colouring. 
These results imply that either the oriented chromatic number for the family of orientations of connected cubic graphs equals the $2$-dipath chromatic number
or the long-standing conjecture of Sopena [\emph{Journal of Graph Theory 25:191-205 1997}] regarding the chromatic number of orientations of connected cubic graphs is false.
\end{abstract}

\section{Introduction and Preliminary Notions}

An \emph{oriented graph} is a simple graph equipped with an orientation of its edges as arcs.
Equivalently, an oriented graph is an antisymmetric loopless digraph.
We call an oriented graph $G$ \emph{properly subcubic} when its underlying simple graph, denoted $U(G)$, has maximum degree three and a vertex of degree at most two.
The \emph{degree} of a vertex in $G$ is its degree in $U(G)$.
We call an oriented graph $G$ \emph{connected} when $U(G)$ is connected.
For $uv, vw \in A(G)$, we say $uvw$ is a \emph{$2$-dipath}; $v$ is \emph{between} $u$ and $w$; and $v$ is the \emph{centre} of the $2$-dipath $uvw$.
For graph theoretic notation and terminology not defined herein, we refer the reader to \cite{bondy2008graph}.

Let $G$ and $H$ be oriented graphs.
There is a \emph{homomorphism of $G$ to $H$} when there exists $\phi:V(G) \to V(H)$ so that $uv \in A(G)$ implies $\phi(u)\phi(v) \in A(H)$.
We call $\phi$ a \emph{homomorphism} or an \emph{oriented $|V(H)|$-colouring of $G$}. 
The  \emph{oriented chromatic number} of an oriented graph $G$, denoted $\chi_o(G)$, the least integer $k$ such that there is a homomorphism of $G$ to an oriented graph with $k$ vertices.
That is, $\chi_o(G)$ is the least integer $k$ so that there exists  an oriented $k$-colouring of $G$.
For $\mathcal{F}$, a  family of oriented graphs,  we define $\chi_o(\mathcal{F})$  to be the least integer $k$ such that $\chi_o(F) \leq k$ for all $F \in \mathcal{F}$.

Equivalently one may define an \emph{oriented $k$-colouring} as a labelling $c:V(G) \to \{0,1,\dots,~k~-~1\}$ such that:
\begin{enumerate}
	\item $c(u) \neq c(v)$ for all $uv \in A(G)$; and
	\item if $uv, xy \in A(G)$ and $c(u) = c(y)$, then $c(v) \neq c(x)$.
\end{enumerate}

Such a labeling implicitly defines a homomorphism to an oriented graph $H$ with vertex set $\{0,1,\dots, k-1\}$ where $ij \in A(H)$ when there is an arc $uv \in A(G)$ such that $\phi(u)=i$ and $\phi(v) = j$.
Condition (1) ensures $H$ is loopless and condition (2) ensures $H$ is antisymmetric. 
 In the case $v = x$, condition (2)  implies that in an oriented colouring any three vertices forming a directed path (i.e., a $2$-dipath) receive distinct colours.
As such one may bound the oriented chromatic number of an oriented graph by considering colourings that assign distinct colours to vertices at directed distance at most two.

Let $G$ be an oriented graph.  A \emph{2-dipath $t$-colouring} of $G$ is a labelling $c:V(G) \to \{0,1,\dots, t-1\}$ such that:
\begin{enumerate}
	\item $c(u) \neq c(v)$ for all $uv \in A(G)$; and
	\item if $uv, vw \in A(G)$, then $c(u) \neq c(w)$.
\end{enumerate}

The \emph{$2$-dipath chromatic number} of an oriented graph $G$, denoted $\chi_{2d}(G)$, is the least integer $t$ such that $G$ admits a $2$-dipath $t$-colouring.
For $\mathcal{F}$, a family of oriented graphs, $\chi_{2d}(\mathcal{F})$ is defined to be the least integer $t$ such that $\chi_{2d}(F) \leq t$ for all $F \in \mathcal{F}$. Since every oriented colouring is $2$-dipath colouring it follows that for any oriented graph $H$ we have $\chi_o(H) \geq \chi_{2d}(H)$.

Consider an oriented graph $G$ for which every pair of vertices is either adjacent or the ends of a  $2$-dipath.
The definitions of oriented colouring and $2$-dipath coloring imply $\chi_o(G) = \chi_{2d}(G) = |V(G)|$.
We call such oriented graphs  \emph{oriented cliques}.

The notion of $2$-dipath colourings was  introduced in  \cite{MAYO12,KMY09}.
We refer the reader to \cite{S16} for a comprehensive survey on homomorphism and colourings of oriented graphs.

Let $\mathcal{F}_3$ be the set of orientations of cubic graphs.
Let $\mathcal{F}^C_3$ be the set of orientations of connected cubic graphs.
In \cite{D19} the authors show $\chi_o(\mathcal{F}^c_3) \leq 9$.
They further show that this bound may be improved to $8$ when restricted to those  oriented connected cubic graphs that have a source or a sink.
Central to these results are homomorphisms to a particular class of Cayley digraphs, namely Paley tournaments.
Let $q$ be a prime power congruent to $3$ modulo $4$. The \emph{Paley tournament on $q$ vertices}, denoted $QR_q$, is the tournament with vertex set $\{0,1,2,\dots, q-1\}$, where $uv \in A(G)$ when $v-u$ is a non-zero quadratic residue modulo $q$.

\begin{theorem}\label{thm:subCubic}\cite{D19}
	If $G$ is a connected properly subcubic oriented graph with no degree $3$ source vertex adjacent to a degree $3$ sink vertex, then $G \to QR_7$.
\end{theorem}

Our work proceeds as follows.
In Section \ref{sec:OrientColourings} we show that if $G$ is an orientation of a connected cubic graph, then $\chi_o(G) \leq 8$, regardless of the presence of sources and sinks.
This improves the best known upper bound for the oriented chromatic number of orientations of connected cubic graphs.
In Section \ref{sec:2DPColorings} we show  that every oriented cubic graph admits a $2$-dipath $7$-colouring.
Together these results point to a new line of attack on the long-standing conjecture of Sopena \cite{SO97} regarding the oriented chromatic number of connected cubic graphs.
This is discussed further in Section \ref{sec:Discussion}.

\section{Oriented Colourings of Orientations of Connected Cubic Graphs}\label{sec:OrientColourings}

For convenience we provide the following result regarding $QR_7$.
\begin{lemma}\label{lem:QR7} \cite{M07}
	\begin{enumerate}
		\item The tournament $QR_7$ is vertex transitive and arc transitive. 
		\item For $yz \in A(QR_7)$,  there exists a pair of distinct vertices $x,x^\prime \in V(QR_7)$ so that $xyz$ and $x^\prime yz$ are directed cycles in $QR_7$.
	\end{enumerate}
\end{lemma}

We begin with technical lemma.

\begin{lemma}\label{lem:AltPath}
	If $G$ is an oriented cubic graph with no source and no sink, then 
	\begin{itemize}
		\item $G$ contains a vertex with out-degree $2$ whose out-neighbours both have in-degree $2$; or
		\item $G$ contains a vertex of out-degree $2$ that has an out-neighbour of in-degree $2$ and an in-neighbour of out-degree $2$.
	\end{itemize}
\end{lemma}

\begin{proof}
	Without loss of generality  assume $G$ is connected.
	We partition the vertices of $G$ based on their out-degree. Let $V^+$ be the set of vertices of $G$ with out-degree $2$ and $V^-$ be the set of vertices in $G$ with in-degree $2$.
	If $G[V^+]$ has a vertex with out-degree $0$, then such a vertex has two out-neighbours in $V^{-}$ and so is a vertex with out-degree $2$ whose out-neighbours both have in-degree $2$.
	Otherwise, assume every vertex in $V^+$ has out-degree at least $1$ in $G[V^+]$.
	

	As every vertex in $G[V^+]$ has out-degree at least $1$ in  $G[V^+]$, the oriented graph $G[V^+]$ contains at least one directed cycle, $C$.
	Note that such a cycle is necessarily induced;
	the head of a chord in such a cycle would have in-degree $2$.
	Let $x$ be a vertex of this cycle, and consider $xx^\prime \in A(G)$ such that $x^\prime$ is not contained in $C$.
	If $x^\prime \in V^{-}$, then $x$ is a vertex of out-degree $2$ that has an out-neighbour of in-degree $2$ and an in-neighbour of out-degree $2$.
	And so assume $x^\prime \in V^{+}$.
	
	Consider a maximal directed walk $W$ in $G$ beginning with $xx^\prime$ so that at most a single vertex of $G$ appears twice in $W$.
	If $W$ contains a vertex of in-degree $2$, then the first arc of $W$ that has its tail in $V^+$ and its head in $V^-$ contains  a vertex of out-degree $2$ that has an out-neighbour of in-degree $2$ and an in-neighbour of out-degree $2$.
	So assume such a walk contains only vertices with out-degree $2$.
	That is, $W$ is contained wholly in $G[V^+]$.
	
	As  each vertex in $G[V^+]$ has positive out-degree and $W$ is maximal, $W$ contains a directed cycle $C^\prime$.
	We claim  $C$ and $C^\prime$ contain no common vertices.
	If  $y \in C$ and $y \in C^\prime$ then by construction there is a directed path $Q$ from $x$ to $y$ in $G[V^+]$ that begins with the arc $xx^\prime$.
	Since $x^\prime \notin C$ and $y \in C$ there is a first vertex $y^\prime \in Q$ such that the in-neighbour of $y^\prime$ is in $Q$ but not in $C$.
	However $y^\prime$ also has an in-neighbour in $C$.
	This contradicts that $y^\prime \in V^+$
	Therefore $C$ and $C^\prime$ contain no common vertices.
		
	As $W$ begins with $x$ and contains a vertex from $C^\prime$, there is a directed path $P$ (which is a subwalk of $W$) from $x$ to a vertex in $C^\prime$.
	The last vertex on this path, i.e., the first one contained in $C^\prime$, has both an in-neighbour in $P$ and an in-neighbour in $C^\prime$. 
	Such a vertex is contained in $V^-$.
	This contradicts that $W$ is wholly contained within $G[V^+]$.
	Thus $W$ contains a vertex of in-degree $2$ and the proof is complete.
\end{proof}

\begin{lemma}\label{lem:notTriFree}
	If $G$ is an oriented connected cubic graph with no source and no sink and $U(G)$ contains a triangle, then $\chi_o(G) \leq 8$.
\end{lemma}

\begin{proof}
	Let $G$ be an oriented connected cubic graph with no source and no sink so that the vertices $u,v,w$ induce a triangle in $U(G)$ and $uv \in A(G)$.
	Without loss of generality, there are two possible orientations: $uvw$ is a directed cycle or $uw,vw\in A(G)$.
	
	By Theorem \ref{thm:subCubic}, there is a homomorphism $\phi: G - uv \to QR_7$.
	If $\phi(u) = \phi(v)$, then modifying $\phi$ to let $\phi(u) = 7$ yields an oriented $8$-colouring of $G$.
	And so we may assume $\phi(u) \neq \phi(v)$.
	As $QR_7$ is arc transitive, we may assume $\phi(v) = 0$ and $\phi(u) = 1$.
	(Note that if $\phi(v) = 1$ and $\phi(u) = 0$, then $\phi: G \to QR_7$).
	Let $u^\prime \neq w$ and $v^\prime \neq w$ so that $u^\prime$ and $v^\prime$ are respectively neighbours of $u$ and $v$.
	If  $uu^\prime \in A(G)$ or $\phi(u^\prime )\neq 0$, then modifying $\phi$ such that $\phi(u) = 7$ yields an oriented $8$-colouring of $G$.
	And so we may assume  $u^\prime u \in A(G)$ and $\phi(u^\prime) = 0$.
	Similarly, we may assume  $vv^\prime \in A(G)$ and $\phi(v^\prime) = 1$.
	Let $w^\prime\notin \{u,v\}$ be a neighbour of $w$. 
	
	\emph{Case I: $uvw$ is a directed cycle in $G$.}
	Consider modifying $\phi$ in one of three ways:
	\begin{enumerate}
		\item $\phi(u) = 2, \phi(v) = 4$;
		\item $\phi(u) = 2, \phi(v) = 6$; or
		\item $\phi(u) = 4, \phi(v) = 6$.
	\end{enumerate}
	Note that in each case $\phi(u)\phi(u^\prime), \phi(v^\prime)\phi(v),\phi(u)\phi(v) \in A(QR_7)$.
	By part (2) of Lemma \ref{lem:QR7}, each of these three possibilities allow us to  modify $\phi(w)$ in two possible ways so that $\phi(w)\phi(u), \phi(v)\phi(w) \in A(QR_7)$:
	\begin{enumerate}
	\item $\phi(w) = 1,5$;
	\item $\phi(w) = 0,1$; or
	\item $\phi(w) = 0,3$.
	\end{enumerate}

Note now that each vertex of $QR_7$  has both an in-neighbour and an out-neighbour in the set $\{0,1,3,5\}$.
As such, regardless of the orientation of the edge $ww^\prime$ and the value of $\phi(w^\prime)$, we can choose $\phi(u),\phi(v)$ and $\phi(w)$ so that $\phi:G \to QR_7$ is a homomorphism.
This completes Case I.
 
\emph{Case II: $uw,vw\in A(G)$.}
Since $G$ has no source or sink vertex, $ww^\prime \in V(G)$.
Proceeding as in Case I, we note that  $\phi(u)$ and $\phi(v)$ can be modified so that $\phi(w^\prime) \notin \{\phi(u),\phi(v)\}$ and $\phi(u)\phi(v) \in A(QR_7)$.
Modify $\phi$ so that $ \phi(u) \neq \phi(w^\prime), \phi(v) \neq \phi(w^\prime)$ and $\phi(w) = 7$.
This yields an oriented $8$-colouring of $G$.
\end{proof}

\begin{lemma}\label{lem:triFree}
	If $G$ is an oriented connected cubic graph with no source and no sink and $U(G)$ is triangle free, then $\chi_o(G)\leq 8$.
\end{lemma}

\begin{proof}
	Let $G$ be an oriented connected cubic graph with no source and no sink so that $U(G)$ is triangle free.
	Consider first the case that $G$ has a cut arc, say $uv$.
	Let $G_u$ be the component of $G-uv$ that contains $u$.
	Similarly, let $G_v$ be the component of $G-uv$ that contains $v$.
	By Theorem \ref{thm:subCubic}, each of $G_u$ and $G_v$ admit a homomorphism to $QR_7$.
	Further, by Lemma \ref{lem:QR7} there exist homomorphisms $\phi_u: G_u \to QR_7$ and $\phi_v: G_v \to QR_7$ so that $\phi_u(u)=0$ and $\phi_v(v) = 1$.
	Combining $\phi_u$ and $\phi_v$ yields a homomorphism of $G$ to $QR_7$. 
	Thus $\chi_o(G)\leq 7$.
	
	Assume now that $G$ has no cut arc.
	By Lemma \ref{lem:AltPath}, in $G$ there is an arc from a vertex of out-degree $2$ to a vertex of in-degree $2$.
	Let $x$ and $u$ be such vertices so that $xu \in A(G)$.
	Let $v$ be the out-neighbour of $u$.
	Let $w\neq x$ be an in-neighbour of $u$.
	Let $z$ be the in-neighbour of $x$.
	Let $y \neq u$ be an out-neighbour of $x$.
	By Lemma \ref{lem:AltPath} we may choose $x$ and $u$ so that $z$ has out-degree $2$ or $y$ has in-degree $2$.
	Note that as $U(G)$ is triangle free, $z$ and $y$ are not adjacent.
	Construct $G^\prime$ from $G$ by removing $x$ and adding the arc $yz$.
	The oriented graph $G^\prime$ is properly subcubic.
	We further note that as $xu$ is not a cut arc, the oriented graph $G^\prime$ is connected.
	
	Recall that in $G$ vertex $z$ has out-degree $2$ or $y$ has in-degree $2$.
	If $y$ is a source vertex in $G^\prime$, then $y$ has out-degree $2$ in $G$.
	Therefore $z$ has out-degree $2$ in $G$ and so is not a sink vertex in $G^\prime$.
	Similarly, if $z$ is a sink vertex in $G^\prime$, then $y$ is not a source vertex in $G^\prime$.
	As $G$ has no source vertex and no sink vertex, exactly one of the following is true: (1) $G^\prime$ has no source or sink vertex of degree $3$; (2) $G^\prime$ has a sink vertex of degree $3$, namely $z$, and no source vertex of degree $3$, or (3) $G^\prime$ has a source vertex of degree $3$, namely $y$, and no sink vertex of degree $3$.
	Thus $G^\prime$, has no degree $3$ source vertex adjacent to a degree $3$ sink vertex.
	By Theorem \ref{thm:subCubic}, there exists a homomorphism $\phi: G^\prime \to QR_7$.
	As $QR_7$ is vertex transitive, we may assume $\phi(v) = 0$.
	By part (2) of Lemma \ref{lem:QR7}, $\phi$ can be extended to include $x$ so that $\phi(x) \neq 0$.
	Note that $\phi(w) \neq 0$ as there is a $2$-dipath from $w$ to $v$ and $\phi(v) = 0$.
	Recoloring $u$ so that $\phi(u) = 7$ gives an oriented $8$-colouring of $G$.

\end{proof}

\begin{theorem} \label{thm:main}
	If $G$ is an oriented connected cubic graph, then $\chi_o(G) \leq 8$.
\end{theorem}

\begin{proof}
	If $G$ has a source or a sink vertex, then the result follows from Corollary 4.9 in \cite{D19}.
	If $G$ has no source and no sink and $U(G)$ contains a triangle, then the result follows by Lemma \ref{lem:notTriFree}.
	Otherwise, $G$ has no source and no sink and $U(G)$ contains no triangle. 
	The result then follows by Lemma \ref{lem:triFree}.
\end{proof}

Figure 1 in \cite{D19} gives an oriented clique on $7$ vertices whose underlying graph has maximum degree $3$.
Thus  $\chi_o(\mathcal{F}^C_3) \geq 7$.
Combining this with the statement of Theorem \ref{thm:main} yields the following.

\begin{theorem}
	For $\mathcal{F}^C_3$, the family of orientations of connected cubic graphs, we have $7 \leq \chi_o(\mathcal{F}^C_3) \leq 8$.
\end{theorem}

\section{$2$-dipath Colourings of Orientations of Cubic Graphs}\label{sec:2DPColorings}

For an oriented graph $G$, let $G^2$ be the simple undirected graph formed from $G$ by first adding an edge between any pair of vertices at directed distance exactly $2$ in $G$ (i.e., vertices at the end of an induced $2$-dipath) and then changing all arcs to edges.
One easily observes $\chi(G^2) = \chi_{2d}(G)$.
Thus we approach our study of $\chi_{2d}(\mathcal{F}_3)$ by examining the chromatic number of graphs of the form $G^2$ for $G \in \mathcal{F}_3$.
Let $\mathcal{F}_3^2 = \{ G^2 | G \in \mathcal{F}_3 \}$.
In \cite{D18} the authors establish $\omega(\mathcal{F}_3^2) = 7$.
Thus $\chi_{2d}(\mathcal{F}_3) \geq 7$.
Here we show $\chi_{2d}(\mathcal{F}_3) = 7$.

\begin{lemma}\label{lem:deg6}
	If $G$ is an orientation of a cubic graph, then $G^2$ is $7$-regular or $G^2$ has a vertex of degree at most $6$.
\end{lemma}

\begin{proof}
	Let $G$ be an orientation of a cubic graph  with $n$ vertices.
	We show  $G^2$ has average degree $7$.
	Every edge in $G^2$ corresponds to  an arc in $G$ or to an induced $2$-dipath in $G$.
	Every vertex in $G$ is the centre vertex of at most two induced $2$-dipaths.
	Therefore there are at most $2n$ induced $2$-dipaths in $G$.
	And so $|E(G^2)| \leq \frac{3n}{2} + 2n  = \frac{7n}{2}$.
	Thus $G^2$ has average degree $7$.
\end{proof}

\begin{lemma}\label{lem:notComplete}
		If $G$ is an orientation of a cubic graph and $G^2$ is $7$-regular, then $\chi(G^2) \leq 7$.
\end{lemma}

\begin{proof}
	By Brooks' Theorem, it suffices to show  $G^2$ is not a complete graph.
	If $G^2$ is a complete graph, then it an oriented clique with at least $8$ vertices.
	However this contradicts the statement of Proposition 3.3 in \cite{D19}.
\end{proof}

\begin{lemma}\label{lem:noadjst}
	If $ G$ is an orientation of a connected cubic graph and $G$ has a source vertex adjacent to sink vertex, then $\chi_{2d}(G) \leq 7$.
\end{lemma}

\begin{proof}
	Let $G$ be an orientation of a connected cubic graph. 
	Let $s_1,s_2,\dots s_\ell$ and $t_1,t_2,\dots, t_\ell$ respectively be source and sink vertices so that $s_it_i \in A(G)$ for all $1 \leq i \leq \ell$.
	(Note that $s_1,s_2, \dots, s_\ell$ may not all be distinct vertices. Similarly $t_1,t_2,\dots, t_\ell$  may not all be distinct vertices.)
	Form $G^\prime$ from $G$ by first deleting each arc $s_it_i$ and then adding vertex $x_i$ so that $s_ix_i, x_it_i \in A(G^\prime)$.
	Notice that for any pair $u,v \in V(G)$, if $u$ and $v$ are at directed distance at most $2$ in $G$, then $u$ and $v$ are at directed distance at most $2$ in $G^\prime$.
	Therefore $\chi_{2d}(G) \leq \chi_{2d}(G^\prime)$.
	The oriented graph $G^\prime$ is a connected properly subcubic oriented graph with no degree $3$ source adjacent to a degree $3$ sink. 
	And so by Theorem \ref{thm:subCubic}, it follows $G^\prime \to QR_7$.
	Thus $\chi_o(G^\prime) \leq 7$ and so $\chi_{2d}(G) \leq 7$.
\end{proof}

\begin{lemma}\label{lem:adjst}
If $G$ is an orientation of a connected cubic graph and $G$ with no source vertex adjacent to sink vertex, then $\chi_{2d}(G) \leq 7$.
\end{lemma}

\begin{proof}
	We proceed by contradiction.
	Consider, $G$, an orientation of a connected cubic graph with no source vertex adjacent to sink vertex so that $\chi_{2d}(G) > 7$.
	Clearly $G \not\to QR_7$, as otherwise $\chi_o(G) \leq 7$.
	However we note that by Theorem \ref{thm:subCubic}, every proper subgraph of $G$ admits a homomorphism to $QR_7$.

	If $G^2$ is $7$-regular, then the result follows by Lemma \ref{lem:notComplete}.	
	If $G^2$ is $6$-degenerate, then $\chi(G^2) \leq 7$.
	And so we may assume that $G^2$ is not $6$-degenerate nor $7$-regular.
	Therefore $G^2$ contains an induced subgraph with minimum degree $7$.
	Let $C$ be the vertices of a maximum such induced subgraph of $G^2$.
	Let $H = V(G) \setminus C$.
	By Lemma \ref{lem:deg6}, $G^2$ has a vertex of degree at most $6$ and so $C \neq V(G)$ and  $H \neq \emptyset$.

	Since $C$ is maximum, there is an ordering of the elements of $H: x_1,x_2,x_3, \dots, x_\ell$ so that for all $1 \leq i \leq \ell$ vertex $x_i$ has degree at most $6$ in the subgraph of $G^2$ induced by the vertices of $C$ together with $x_1,x_2,\dots, x_{i-1}$.
	We see then that if $G^2[C]$ admits a  $7$-colouring, such a colouring can be extended to a $7$-colouring of $G^2$.
	As this would be a contradiction, we may assume $\chi(G^2[C]) > 7$.
	
	Edges in $G^2[C]$ arise from arcs in $G[C]$ and from $2$-dipaths in $G$ whose ends are in $C$.
	For these latter edges, it is possible the the centre vertex of such a $2$-dipath is not contained in $C$.
	Let $B_C = \{w\in H | w \mbox{ is between a pair of vertices in $C$ }\}$.
	That is, $B_C$ is the set of vertices of $G$ that are not in $C$ and are the centre vertex of a $2$-dipath in$G$ whose ends are in $C$.  
	Consider the subgraph of $G$, $G_C$, formed from $G[C]$ by adding the vertices of $B_C$ and all arcs with exactly one endpoint in $C$ and one end point in $B_C$.
	Since $U(G)$ is cubic, every proper subgraph $U(G)$ has components that are properly subcubic.
	Thus, if $G_C$ is a proper subgraph of $G$ (i.e, if $G_C \neq G$), then every component of $G_C$ is a connected properly subcubic oriented graph with no degree $3$ source vertex adjacent to a degree $3$ sink vertex.
	By Theorem \ref{thm:subCubic} there exists a homomorphism $\phi: G_C \to QR_7$.
	Restricting $\phi$ to the vertices of $C$ yields an $7$-colouring of $G^2[C]$.
	This is a contradiction as $\chi(G^2[C]) > 7$.
	Therefore $G_C = G$.
	
	Since $G_C = G$ it follows $B_C = H$ and $G[H]$ is an independent set.
	We proceed to bound $e_C$, the number of edges in $G^2[C]$.
	Recall that edges in $G^2[C]$ arise from arcs in $G[C]$ and from induced $2$-dipaths in $G$, $xyz$, so that $x,z\in C$.

	The oriented graph $G$ is cubic and so has $\frac{3(|C|+|H|)}{2}$ arcs.
	Of these, exactly $3|H|$ arcs have exactly one end in $H$.
	Therefore $G^2[C]$ has at most $\frac{3(|C|+|H|)}{2}$ - $3|H|$ edges that arise from arcs in $G[C]$.
	
	Since $V(G) = C \cup H$, the set of induced $2$-dipaths $xyz \in G$ with $x,z\in C$ can be partitioned into those for which $y \in H$ and those for which $y \in C$.

	Each vertex of $H$ is the centre vertex of at most two induced  $2$-dipaths in $G$.
	Therefore there are at most $2|H|$ induced $2$-dipaths  $xyz$ in $G$ with $x,z\in C$ so that $y \in H$.

	Since each vertex of $C$ has degree at least $7$ in $G^2[C]$ and each vertex of $G$ has at most six vertices at distance exactly $2$ in $U(G)$, we observe  $G[C] $ has no isolated vertices.
	As each vertex of $C$ has degree at most $3$ and at least $1$ in $G[C]$,  each vertex $y \in C$ is the centre vertex of at most $deg_{G[C]}-1\geq 0$ induced $2$-dipaths in $G[C]$.
	Therefore there are at most \[\sum\limits_{y \in C} deg_{G[C]}(y)-1 = \left(\sum\limits_{y \in C} deg_{G[C]}(y)\right) - |C| \]  induced $2$-dipaths in $G[C]$.
	To find the sum of the degrees of the vertices in $G[C]$, we recall $G$ is cubic and so the sum of the degrees in $G$ is $3(|C| + |H|)$.
	Discarding a vertex $H$ from this count decreases this sum by exactly $6$; each vertex in $H$ is incident with $3$ arcs and each of these arcs has an endpoint in $C$. 
	Therefore $\sum\limits_{y \in C} deg_{G[C]}(y)  = 3(|C| + |H|) -6|H|$.
	And so there are at most 
	\[3(|C| + |H|) -6|H| - |C| = 2|C| - 3|H|\] induced $2$-dipaths  in $G[C]$.
	As  such there are at most $2|C| - 3|H|$  induced $2$-dipaths  $xyz$ in $G$ with $x,z\in C$ so that $y \in C$.
	
	Therefore \[e_C \leq  \frac{3(|C|+|H|)}{2} - 3|H| + 2|H| +2|C| - 3|H| = \frac{7|C|}{2} - \frac{5|H|}{2}.\]
	Since $H$ is non-empty, we conclude $e_C < \frac{7|C|}{2}$.
	On the other hand, we recall $G^2[C]$ has minimum degree $7$. 
	And so $e_C \geq \frac{7|C|}{2}$, a contradiction.

%
%
 
\end{proof}

\begin{theorem}\label{thm:main2}
	If $G$ is an orientation of a cubic graph, then $\chi_{2d}(G) \leq 7$.
\end{theorem}

\begin{proof}
	It suffices to assume $G$ is connected. The result follows directly from  Lemmas \ref{lem:noadjst} and \ref{lem:adjst}.
\end{proof}

\section{Discussion}\label{sec:Discussion}

In an early investigation into the oriented chromatic number of orientations of connected cubic graphs Sopena \cite{SO97} conjectured $\chi_o(\mathcal{F}^C_3) = 7$.
This conjecture has been verified for all orientations of connected cubic graphs with fewer than $20$ vertices \cite{P06}.
Here we have shown $\chi_o(\mathcal{F}^C_3) \in \{7,8\}$ and  $\chi_{2d}(\mathcal{F}_3) = 7$.
Our results imply exactly one of the following must be true:
\begin{enumerate}
	\item $\chi_o(\mathcal{F}^C_3) = 8$ or
	\item  $\chi_o(\mathcal{F}^C_3) = \chi_{2d}(\mathcal{F}_3) = 7$.
\end{enumerate} 

In other words, either Sopena's conjecture is false or any orientation of a cubic graph with $2$-dipath chromatic number $7$ also has oriented chromatic number $7$.
Notably, a statement analogous to (2) is true for orientations of $2$-regular graphs.
To wit, $\chi_o(\mathcal{F}^C_2) = \chi_{2d}(\mathcal{F}_2) = 5$.
(In fact $\chi_o(\mathcal{F}_2) = \chi_{2d}(\mathcal{F}_2) = 5$ \cite{S16}.)
It remains to be seen if this is an artefact of the simple structure of $2$-regular graphs or due to some deeper connection between oriented colourings and $2$-dipath colourings of orientations of bounded degree graphs.
In either case, the study of oriented graphs that have oriented chromatic number equal to $2$-dipath chromatic number presents a new line of attack for this long-standing open problem.

In \cite{MAYO12} the authors define, for each $k > 1$, an oriented graph $H_k$ with the property that $G \to H_k$ if and only if $\chi_{2d}(G) \leq k$.
An important next step in the study of homomorphisms of orientations of connected cubic graphs is the a study of the oriented chromatic number of various subgraphs of $H_7$.

Our results for the oriented chromatic number are limited to orientations of cubic graphs that are connected. 
Recent work by Dybizba\'nskia, Ochem, Pinlou and Szepietowskia \cite{O19} extend the work in \cite{D19} and show that all oriented cubic graphs admit an oriented $9$-colouring.
This result follows from construction of a $9$-vertex universal target, based on $QR_7$, for the family of orientations of cubic graphs.
Using similar  methods, these authors provide new upper bounds for the oriented chromatic number of other families of orientations of bounded degree  graphs.
It remains to be seen if Theorem \ref{thm:main} can be extended in a similar manner.


\bibliographystyle{plain}
\bibliography{references}

\end{document}